\documentclass[10pt, conference, a4paper]{IEEEtran}
\ifCLASSINFOpdf
\else
\fi
\usepackage{lipsum}
\usepackage{graphicx}
\usepackage{float}
\usepackage{cite}
\usepackage{epstopdf}
\usepackage{multirow}
\usepackage{amssymb}
\usepackage{amsmath}
\usepackage{amsthm}
\usepackage{mathtools, cases}
\usepackage{chngcntr}
\usepackage{amsmath,etoolbox}
\usepackage{amssymb}
\usepackage[ruled]{algorithm2e}
\usepackage{url}
\usepackage{caption}
\usepackage{subcaption}
\usepackage[font=small]{caption}
\usepackage{cleveref}
\usepackage{color}
\usepackage[utf8]{inputenc}
\usepackage[english]{babel}
\newtheorem{theorem}{Theorem}
\newtheorem{lemma}{Lemma}
\makeatletter
\renewcommand\footnoterule{%
  \kern-3\p@
  \hrule\@width.4\columnwidth
  \kern2.6\p@}
  \makeatother

\newcommand{\E}{\mathbb{E}}

\newcommand{\setC}{\mathcal{C}}
\begin{document}
%
\title{Limited by Capacity or Blockage?\\ A Millimeter Wave Blockage Analysis\vspace{-3mm}}
\author{Ish Kumar Jain, Rajeev Kumar, and Shivendra Panwar\\
Department of Electrical and Computer Engineering, Tandon School of Engineering, NYU, NY 11201, USA\\
}
\maketitle

\begin{abstract}

Millimeter wave (mmWave) communication systems can provide high data rates but the system performance may degrade significantly due to mobile blockers and the user's own body.
A high frequency of interruptions and long duration of blockage may degrade the quality of experience. For example, delays of more than about 10ms cause nausea to VR viewers.
Macro-diversity of base stations (BSs) has been considered a promising solution where the user equipment (UE) can handover to other available BSs, if the current serving BS gets blocked.
However, an analytical model for the frequency and duration of dynamic blockage events in this setting is largely unknown.
In this paper, we consider an open park-like scenario and obtain closed-form expressions for the blockage probability, expected frequency and duration of blockage events using stochastic geometry. Our results indicate that the minimum density of BS that is required to satisfy the Quality of Service (QoS) requirements of AR/VR and other low latency applications is largely driven by blockage events rather than capacity requirements. Placing the BS at a greater height reduces the likelihood of blockage. We present a closed-form expression for the BS density-height trade-off that can be used for network planning.
\end{abstract}
\IEEEpeerreviewmaketitle

\section{Introduction}
\label{sec:intro}
Millimeter wave (mmWave) communication systems can provide high data rates of the order of a few Gbps~\cite{CellularCap-Rap}, suitable for the  Quality of Service (QoS) requirements for Augmented Reality (AR) and Virtual Reality (VR). For these applications, the user-equipment (UE) requires the data rate to be in the range of 100 Mbps to a few Gbps, and an end-to-end latency in the range of 1 ms to 10 ms~\cite{ATTARVR}. However, mmWave communication systems are quite vulnerable to blockages due to higher penetration losses and reduced diffraction~\cite{bai2015coverage}. Even the human body can reduce the signal strength by 20 dB~\cite{georgeFading}. Thus, an unblocked Line of Sight (LOS) link is highly desirable for mmWave systems. Furthermore, a mobile human blocker can block the LOS path between User Equipment (UE) and Base Station (BS) for approximately 500 ms~\cite{georgeFading}. The frequent blockages of mmWave LOS links and a high blockage duration can be detrimental to ultra-reliable and ultra-low latency applications.

One potential solution to blockages in the mmWave cellular network can be macro-diversity of BSs and coordinated multipoint (CoMP) techniques. These techniques have shown a significant reduction in interference and improvement in reliability, coverage, and capacity in the current Long Term Evolution-Advance (LTE-A) deployments and other communication networks~\cite{kim2011analysis}. Furthermore, Radio Access Networks (RANs) are moving towards the cloud-RAN architecture that implements macro-diversity and CoMP techniques by pooling a large number of BSs in a single centralized baseband unit (BBU)~\cite{IBMCloudRAN,chen2011c}. As a single centralized BBU handles multiple BSs, the handover and beam-steering time can be reduced significantly~\cite{cloudRAN}.  In order to provide seamless connectivity for ultra-reliable and ultra-low latency applications, the proposed 5G mmWave cellular architecture needs to consider key QoS parameters such as the probability of blockage events, the frequency of blockages, and the blockage duration. 
For instance, to satisfy the QoS requirement of mission-critical applications such as AR/VR, tactile Internet, and eHealth applications, 5G cellular networks target a service reliability of 99.999\% \cite{mohr20165g}.  
In general, the service interruption due to blockage events can be alleviated by caching the downlink contents at the BSs or the network edge~\cite{bastug2017toward}. However, caching the content more than about 10ms may degrade the user experience and may cause nausea to the users particularly for AR applications~\cite{westphal2017challenges}. 
An alternative when blocked is to offload traffic to sub-6GHz networks such as 4G, but this needs to be carefully engineered so as to not overload them.
Therefore, it is important to study the blockage probability, blockage frequency, and blockage duration to satisfy the desired QoS requirements. 


This work presents a simple blockage model for the LOS link using tools from stochastic geometry. In particular, our contributions are as follow:
\begin{enumerate}
\item We provide an analytical model for dynamic blockage (UE blocked by mobile blockers) and self-blockage (UE blocked by the user's own body). The expression for the rate of blockage of LOS link is evaluated as a function of the blocker density, velocity, height and link length.

\item We evaluate the closed-form probability and expected frequency of simultaneous blockage of all BSs in the range of the UE. Further, we present an approximation for the expected duration of simultaneous blockage.  
\item We verify our analytical results through Monte-Carlo simulations by considering a random way-point mobility model for blockers.
\item Finally, we present a case study to find the minimum required BS density for specific mission-critical services and analyze the trade-off between BS height and density to satisfy the QoS requirements.
\end{enumerate}

The rest of the paper is organized as follows. The related work is presented in Section \ref{sec:related-work}, and the system model is described in Section \ref{sec:system-model}. Section~\ref{sec:analysis} provides an analysis of blockage events and evaluates the key blockage metrics. The validation of our theoretical results with MATLAB simulations is presented in Section~\ref{sec:numResults}. 
Finally, Section \ref{sec:conclusion} concludes the paper.

\section{Related Work}
\label{sec:related-work}



A mmWave link may have three kinds of blockages, 
%
namely, static, dynamic, and self-blockage. Static blockage due to buildings and permanent structures has been studied in~\cite{bai2012using} and~\cite{bai2014analysis} using random shape theory and a stochastic geometry approach for urban microwave systems. The underlying static blockage model is incorporated into the cellular system coverage and rate analysis in\cite{bai2015coverage}. 
Static blockage may cause permanent blockage of the LOS link. However, for an open area such as a public park, static blockages play a small role.                                      
The second type of blockage is dynamic blockage due to mobile humans and vehicles (collectively called mobile blockers) which may cause frequent interruptions to the LOS link. Dynamic blockage has been given significant importance by 3GPP in TR 38.901 of Release 14~\cite{3gpptr}. An analytical model in~\cite{gapeyenko2016analysis} considers a single access point, a stationary user, and blockers located randomly in an area. The model in~\cite{wang2017blockage} is developed for a specific scenario of a road intersection using a Manhattan Poisson point process model.
MacCartney \textit{et al.}~\cite{georgeFading} developed a Markov model for blockage events based on measurements on a single BS-UE link. Similarly, Raghavan \textit{et al.}~\cite{raghavan2018statistical} fits the blockage measurements with various statistical models. However, a model based on experimental analysis is very specific to the measurement scenario. The authors in~\cite{han20173d} considered a 3D blockage model and analyzed the blocker arrival probability for a single BS-UE pair. 
Studies of spatial correlation and temporal variation in blockage events for a single BS-UE link are presented in~\cite{samuylov2016characterizing} and~\cite{gapeyenko2017temporal}. However, their analytical model is not easily scalable to multiple BSs, important when considering the impact of macro-diversity..

Apart from static and dynamic blockage, self-blockage plays a key role in mmWave performance.
The authors of~\cite{abouelseoud2013effect} studied human body blockage through simulation. A statistical self-blockage model is developed in~\cite{raghavan2018statistical} through experiments considering various modes (landscape or portrait) of hand-held devices. The impact of self-blockage on received signal strength is studied in~\cite{bai2014analysis} through a stochastic geometry model. They assume the self-blockage due to a user's body blocks the BSs in an area represented by a cone.

All the above blockage models consider the UE's association with a single BS. 
Macro-diversity of BSs is considered as a potential solution to alleviate the effect of blockage events in a cellular network.
The authors of~\cite{zhu2009leveraging} and~\cite{zhang2012improving} proposed an architecture for macro-diversity with multiple BSs and showed the  improvement in network throughput. A blockage model with macro-diversity is developed in~\cite{choi2014macro} for independent blocking and in~\cite{gupta2018macrodiversity} for correlated blocking. However, they consider only static blockage due to buildings.

The primary purpose of the blockage models in previous papers was to study the 
coverage and capacity analysis of the mmWave system. However, 
apart from the signal degradation, blockage frequency and duration also affects the performance of the mmWave system and are critically important for applications such as AR/VR.   
In this paper, we present a simple closed-form expression for a compact analysis to provide insight into the optimal density, height and other design parameters and trade-offs of BS deployment.

\section{System Model}
\label{sec:system-model}
Our system model consists of the following settings:
\begin{itemize}
\item \textit{BS Model}: The mmWave BS locations are modeled as a homogeneous Poisson Point Process (PPP) with density $\lambda_T$. 
Consider a disc $B(o,R)$ of radius $R$ and centered around the origin $o$, where a typical UE is located.
We assume that each BS in $B(o,R)$ is a potential serving BS for the UE. 
Thus, the number of BSs $M$ in the disc $B(o,R)$ of area $\pi R^2$ follows a Poisson distribution with parameter $\lambda_T\pi R^2$, \textit{i.e.},
\begin{equation}\label{eqn:poisson}
    P_M(m) = \frac{[\lambda_{T} \pi R^2]^m}{m!}e^{-\lambda_{T} \pi R^2}.
\end{equation}  


Given the number of BSs $m$ in the disc $B(o,R)$, we have a uniform probability distribution for BS locations.
The BSs distances $\{R_i\}$ $\forall i=1,\ldots,m$ from the UE are independent and identically distributed (iid) with distribution
\begin{equation}\label{eqn:distribution}
 f_{R_i|M}(r|m) = \frac{2r}{R^2}; \ 0< r\le R, \forall\ i=1,\ldots,m.
\end{equation}
\begin{figure}[!t]
	\centering
	\includegraphics[width=0.48\textwidth]{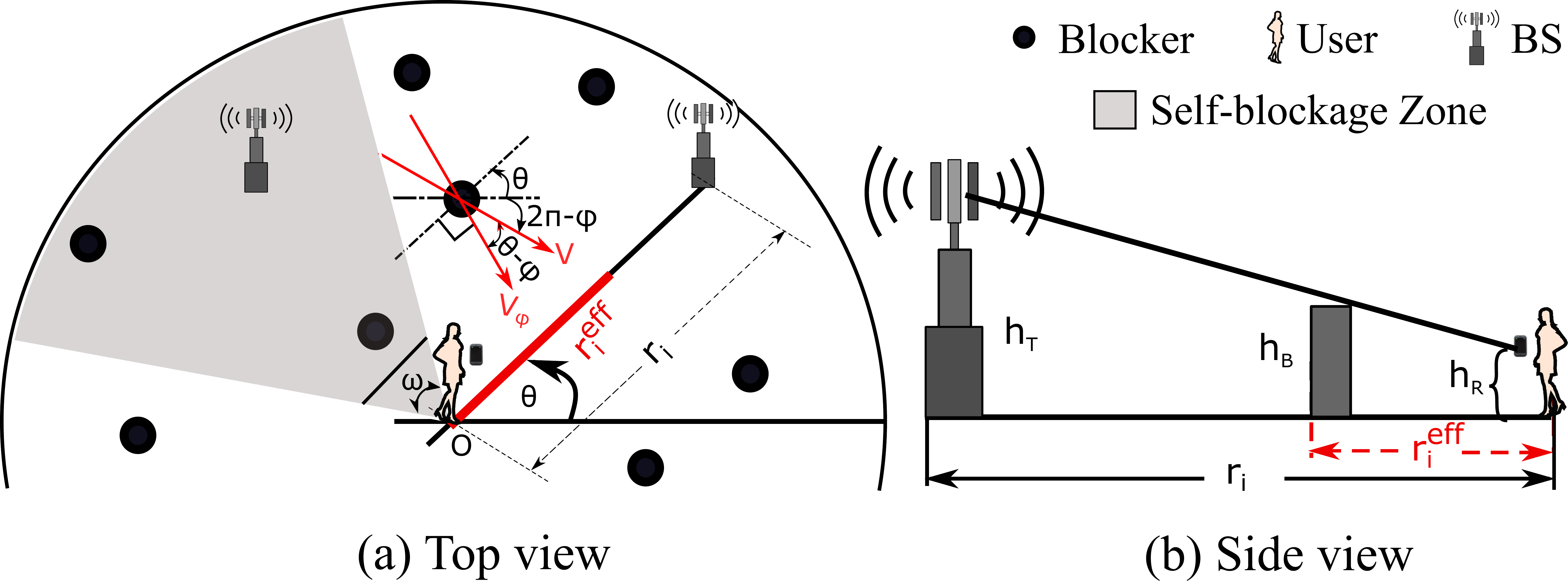}
	\caption{System Model}
	\label{fig:sysMod}
	\vspace{-4mm}
\end{figure}


\item \textit{Self-blockage Model}: The user blocks a fraction of BSs due to his/her own body. The self-blockage zone is defined as a sector of the disc $B(o,R)$ making an angle $\omega$ towards the user's body as shown in Figure~\ref{fig:sysMod}(a). Thus, all of the BSs in the self-blockage zone are considered blocked. 

\item \textit{Dynamic Blockage Model}: 
The blockers are distributed according to a homogeneous PPP with parameter $\lambda_{B}$. Further, the arrival process of the blockers crossing the $i^{th}$ BS-UE link is Poisson with intensity $\alpha_i$. The blockage duration is independent of the blocker arrival process and is exponentially distributed with parameter $\mu$. 



\item \textit{Connectivity Model}: We say the UE is blocked when all of the potential serving BSs in the disc $B(o,R)$ are blocked simultaneously.

\end{itemize}


For a sound understanding of the system model, consider a single BS-UE LOS link in Figure~\ref{fig:sysMod}(a). The distance between the $i^{th}$ BS and the UE is $r_i$ and the LOS link makes an angle $\theta$ with respect to the positive x-axis. Further, the blockers in the region move with constant velocity $V$ at an angle $\varphi$ with the positive x-axis, where $\varphi\sim \text{Unif}([0,2\pi])$. 
Note that only a fraction of blockers crossing the BS-UE link will be blocking the LOS path, as shown in Figure \ref{fig:sysMod}(b). The effective link length $r_i^{eff}$ that is affected by the blocker's movement is 
\begin{equation}
	r_i^{eff}=\frac{\left(h_B-h_R\right)}{\left(h_T-h_R\right)}r_i,
\end{equation}
where $h_B, h_R$, and $h_T$ are the heights of blocker, UE (receiver), and BS (transmitter) respectively. The blocker arrival rate $\alpha_i$ (also called the blockage rate) is evaluated in Lemma \ref{lemma:alphan}. 

\begin{lemma}\label{lemma:alphan}
 The blockage rate $\alpha_i$ of the $i^{th}$ BS-UE link is
\begin{equation}\label{eqn:SingleBS}
\alpha_{i}
=Cr_i,
\end{equation}
where $C$ is proportional to blocker density $\lambda_B$ as,
\begin{equation}\label{eqn:C}
C = \frac{2}{\pi}\lambda_{B} V\frac{(h_B-h_R)}{(h_T-h_R)}.
\end{equation} 

\end{lemma}
\begin{proof}
Consider a blocker moving at an angle $(\theta-\varphi)$ relative to the BS-UE link (See Figure \ref{fig:sysMod}(a)). The component of the blocker’s velocity perpendicular to the BS-UE link is $V_\varphi = V\sin(\theta-\varphi)$, where $V_\varphi$ is positive when $(\theta-\pi)<\varphi<\theta$. 
Next, we consider a rectangle of length $r_i^{eff}$ and width $V_\varphi\Delta t$. The blockers in this area will block the LOS link over the interval of time $\Delta t$. Note there is an equivalent area on the other side of the link. Therefore, the frequency of blockage is $2\lambda_{B}r_i^{eff}V_\varphi\Delta t = 2\lambda_{B} r_i^{eff}V\sin(\theta-\varphi)\Delta t$. Thus, the frequency of blockage per unit time is $2\lambda_{B} r_i^{eff}V\sin(\theta-\varphi)$. Taking an average over the uniform distribution of $\varphi$ (uniform over $[0,2\pi]$), we get the blockage rate $\alpha_i$ as
\begin{equation}\label{eqn:SingleBS}
\begin{split}
\MoveEqLeft\alpha_{i} = 2\lambda_{B} r_i^{eff}V\int_{\varphi = \theta-\pi}^{\theta}\sin(\theta-\varphi)\frac{1}{2\pi}\,d\phi \\
\MoveEqLeft \qquad \quad =\frac{2}{\pi}\lambda_{B} r_i^{eff}V = \frac{2}{\pi}\lambda_{B} V\frac{(h_B-h_R)}{(h_T-h_R)}r_i.
\end{split}
\end{equation}
This concludes the proof. 
\end{proof}
Following~\cite{gapeyenko2017temporal}, we model the blocker arrival process as Poisson with parameter $\alpha_i$ blockers/sec (bl/s).
Note that there can be more than one blocker simultaneously blocking the LOS link. The overall blocking process has been modeled in~\cite{gapeyenko2017temporal} as an alternating renewal process with alternate periods of blocked/unblocked intervals, where the distribution of the blocked interval is obtained as the busy period distribution of a general $M/GI/\infty$ system.
For mathematical simplicity, we assume the blockage duration of a single blocker is exponentially distributed with parameter $\mu$, thus, forming an $M/M/\infty$ queuing system. We further approximate the overall blockage process as an alternating renewal process with exponentially distributed periods of blocked and unblocked intervals with parameters $\alpha_i$ and $\mu$ respectively. This approximation works for a wide range of blocker densities as shown in Section \ref{sec:numResults}.



In the next section, we consider a generalized blockage model for $M$ BSs which are in the range of the UE. The UE keeps track of all these $M$ BSs using beam-tracking and handover techniques. 
Since the handover process is assumed to be very fast, we assume that the UE can instantaneously connect to any unblocked BS when the current serving BS gets blocked. Therefore, we consider the total blockage event occurs only when all the potential serving BSs are blocked. 

\section{Analysis of Blockage Events}\label{sec:analysis}

\subsection{Coverage Probability under Self-blockage}
Let there be $N$ BSs out of total $M$ BSs within the range of the UE that are not blocked by self-blockage. 
The distribution of the number of BSs ($N$) outside the self-blockage zone and in the disc $B(o,R)$ is
\begin{equation}\label{eqn:PN}
P_N(n) = \frac{[p\lambda_{T} \pi R^2]^n}{n!}e^{-p\lambda_{T} \pi R^2}.
\end{equation}
where \begin{equation}\label{eqn:p}
p= 1-\omega/2\pi 
\end{equation} 
is the probability that a randomly chosen BS lies outside the self-blockage zone in the disc $B(o,R)$ (the area of such region would be $p\pi R^2$).
Let $\mathcal{C}$ denotes an event that the UE has at least one serving BS in the disc $B(o,R)$ and outside the self-blockage zone, \textit{i.e.}, $N\ne 0$. The probability of the event $\mathcal{C}$ is called the coverage probability under self-blockage and calculated as,
\begin{equation}\label{eqn:coveragep}
P(\mathcal{C}) =1-P_N(0)=1- e^{-p\lambda_T\pi R^2}.
\end{equation}

\subsection{Blockage Probability}
Given that there are $n$ BSs in the communication range of the UE and are not blocked by the user's body, they can still get blocked by mobile blockers. The blocking event of these $n$ BS-UE links is assumed to be independent on-off processes with $\alpha_i$ and $\mu$ as blocking and unblocking rates, respectively. The probability of blockage of the $i^{th}$ BS-UE link is $\alpha_i/(\alpha_i + \mu)$. Our objective is to develop a blockage model for the mmWave cellular system where UE can connect to any of the potential serving BSs. In this setting the blockage occurs when all the BS-UE links are blocked simultaneously. 
We define an indicator random variable $B$ that indicates the blockage of all available BSs in the range of UE. The blockage probability $P(B|N,\{R_i\})$ is conditioned on the number of BSs $N$ in the disc $B(o,R)$ which are not blocked by the user's body and the link lengths $\{R_i\}$ for $i=1,\cdots,n$.
Since we assume independent blockage of BS-UE links, $P(B|N,\{R_i\})$ can be written as the product of individual blockage probabilities as
\begin{equation}\label{eqn:AllBS}
        P(B|N,\{R_i\})= \prod_{i=1}^n\frac{\alpha_i/\mu}{1+\alpha_i/\mu}   =\prod_{i=1}^n\frac{(C/\mu)r_i}{1+(C/\mu)r_i},
\end{equation} 
where $C$ is defined in (\ref{eqn:C}).
Note that the notation $P(B|N,\{R_i\})$ is a compressed version of $P_{B|N,\{R_i\}}(b|n,\{r_i\})$, which we use for convenience.

To obtain the unconditional blockage probability $P(B)$, we first evaluate the conditional blockage probability $P(B|N)$ by taking the average of $P(B|N,\{R_i\})$ over the distribution of $\{R_i\}$ and then find $P(B)$ by taking the average of $P(B|N)$ over the distribution of $N$ as follow
\begin{equation}\label{eqn:pBgivenN}
\begin{split}
P(B|N)=\!\!\int\!\!\!\int_{r_i}\!\! P(B|N,\{R_i\})\;f(\{R_i\}|N)\; dr_1\cdots dr_n,
 \end{split}
\end{equation}
\begin{equation}\label{eqn:ep1n}
\begin{split}
   P(B)=\sum_{n=0}^{\infty} P(B|N)P_N(n),
\end{split}
\end{equation}

\noindent
\begin{theorem} \label{th1} The marginal blockage probability and the conditional blockage probability conditioned on the coverage event (\ref{eqn:coveragep}) is 

\begin{equation}\label{eqn:expblockage}
        P(B) = e^{-a p\lambda_T\pi R^2},
\end{equation} 
\begin{equation}\label{eqn:expblockagecond}
       P(B|\mathcal{C}) = \frac{e^{-a p\lambda_T\pi R^2} - e^{-p\lambda_T\pi R^2}}{1-e^{-p\lambda_T\pi R^2}},
\end{equation} 
where,
\begin{equation}\label{eqn:a}
\begin{split}
      a=\frac{2\mu}{RC}-\frac{2\mu^2}{R^2C^2} \log\left(1+\frac{RC}{\mu}\right).
   \end{split}
\end{equation} 
Note that $C$ is proportional to blocker density $\lambda_B$ as shown in (\ref{eqn:C}). Also, $p=1-\omega/2\pi$ as defined in (\ref{eqn:p}). 
\end{theorem}
\begin{proof} See Appendix \ref{app:th_prob}.
\end{proof}

For further insight, we can approximate $a$ by taking a series expansion of $\log(1+RC/\mu)$, \textit{i.e.},
\begin{equation}
\begin{split}
a&= \frac{2\mu}{RC}-\frac{2\mu^2}{R^2C^2} \left(\frac{RC}{\mu}-\frac{R^2C^2}{2\mu^2}+\frac{R^3C^3}{3\mu^3}+\cdots\right),\\
&\approx 1-\frac{2RC}{3\mu}, \quad \text{when } \frac{RC}{\mu} \ll 1
\end{split}
\end{equation}
Note that for the blocker density as high as $0.1$ bl/m$^2$, and for other parameters in Table \ref{SimParams}, we have $RC/\mu = 0.35$, which shows that the approximation holds for a wide range of blocker densities.
For large BS density $\lambda_T$, the coverage probability $P(\setC)$ is approximately 1 and $P(B|\setC)\approx P(B)$. 
In order to have the blockage probability $P(B)$ less than a threshold $\bar{P}$
\begin{equation}
P(B)=e^{-ap\lambda_T\pi R^2}\leq\bar{P},
\end{equation}
the required BS density follows
\begin{equation} \label{eqn:approx_lamT}
\lambda_T\geq\frac{-\log(\bar{P})}{ap\pi R^2}\approx \frac{-\log(\bar{P})(1+\frac{2RC}{3\mu})}{p\pi R^2},
\end{equation}
where $C$ is proportional to the blocker density $\lambda_B$ in (\ref{eqn:C}). 
The result (\ref{eqn:approx_lamT}) shows that the BS density approximately scales linearly with the blocker density.




\subsection{Expected Blockage Frequency}
The total arrival rate of blockers in the state when all BSs get simultaneously blocked is same as the total departure rate from that state in equilibrium. Therefore, the frequency/rate of simultaneous blockage of all $N$ BSs is:

\begin{equation}\label{eqn:AllBSfreq}
\begin{split}
    \zeta_B = n\mu P(B|N,\{R_i\})=n\mu \prod_{i=1}^n\frac{(C/\mu)r_i}{1+(C/\mu)r_i},
\end{split}
\end{equation}

The expected rate of blockage is obtained where the expectation is taken over the joint distribution of $N$ and $\{R_i\}$.
\begin{equation}\label{eqn:blockageDurEQn1}
\E[\zeta_B|N] = \!\!\int\!\!\!\int_{r_i}\!\! \zeta_B\; f(\{R_i\}|N)\;dr_1\cdots dr_n,
\end{equation}
\begin{equation}\label{ep1n}
        \mathbb{E} \left[\zeta_B\right] = \sum_{n=0}^\infty \E[\zeta_B|N]P_N(n).
\end{equation}

\begin{theorem} \label{th2} The expected frequency of simultaneous blockage of all BSs in the disc of radius $R$ around UE is 
\begin{equation}\label{eqn:blkfrq}
        \mathbb{E}\left[\zeta_B\right] = \mu (1-a)p \lambda_T\pi R^2 e^{-ap\lambda_T\pi R^2},
\end{equation}       and the expected frequency conditioned on the coverage event (\ref{eqn:C}) is \begin{equation}\label{eqn:blkfreq_cond}
\begin{split}
\MoveEqLeft \mathbb{E}\left[\zeta_B|\setC\right] = \frac{\mu (1-a)p\lambda_T\pi R^2e^{-ap\lambda_T\pi R^2}}{{1-e^{-p\lambda_T\pi R^2}}}.
\end{split}
\end{equation}
where $a$ is defined in (\ref{eqn:a}).
\end{theorem}
\begin{proof}See Appendix \ref{app:th_freq}.
\end{proof}


  

\subsection{Expected Blockage Duration}
Recall that the duration of the blockage of a single BS-UE link is an exponential random variable with mean $\mu$. 
Since the blockage of individual BS-UE links are independent, the duration $T_B$ of the blockage of all $n$ BSs follows an exponential distribution with mean $1/n\mu$.
We can therefore write the expected blockage duration as 
\begin{equation}\label{eqn:exp_TBgivenN}
\E[T_B|N] = \frac{1}{n\mu}
\end{equation}

\noindent
\begin{theorem}\label{th3} The expected blockage duration of simultaneous blockage of all the BSs in $B(o,R)$, conditioned on the coverage event $\setC$ in (\ref{eqn:coveragep}), is obtained as
\begin{equation}\label{eqn:blockageDurationExp}
        \mathbb{E}\left[T_B|\setC\right] = \frac{e^{-p\lambda_T\pi R^2}}{\mu\left(1-e^{-p\lambda_T\pi R^2}\right)}\text{Ei}\left[p\lambda_T\pi R^2\right].
\end{equation}        
where, $\text{E}\text{i}\left[p\lambda_T\pi R^2\right]$ 
= $\sum_{n=1}^\infty\frac{[p\lambda_T\pi R^2]^n}{nn!}$.
\end{theorem}
\begin{proof}
 See Appendix \ref{app:th_dur}
\end{proof}

\noindent
\begin{lemma}
$\text{E}\text{i}\left[p\lambda_T\pi R^2\right]$ converges.
\end{lemma} 
\begin{proof}
 We can use Cauchy ratio test to show that the series $\sum_{n=1}^\infty\frac{[p\lambda_T\pi R^2]^n}{nn!}$ is convergent. Consider 
$L = \lim_{n\rightarrow\infty}\frac{[p\lambda_T\pi R^2]^{n+1}/((n+1)(n+1)!)}{[p\lambda_T\pi R^2]^n/(nn!)} = \lim_{n\rightarrow\infty}\frac{[p\lambda_T\pi R^2]n}{(n+1)^2}=0$. Hence, the series converges.
\end{proof}
An approximation of blockage duration can be obtained for a high BS density as follow
\begin{equation}
\E[T_B|\setC]\approx \frac{1}{\mu p\lambda_T\pi R^2}+ \frac{1}{(\mu p\lambda_T\pi R^2)^2}.
\end{equation}
This approximation is justified in Appendix \ref{app:approx_dur}.

\begin{figure}[!t]
	\centering
	\includegraphics[width=0.435\textwidth]{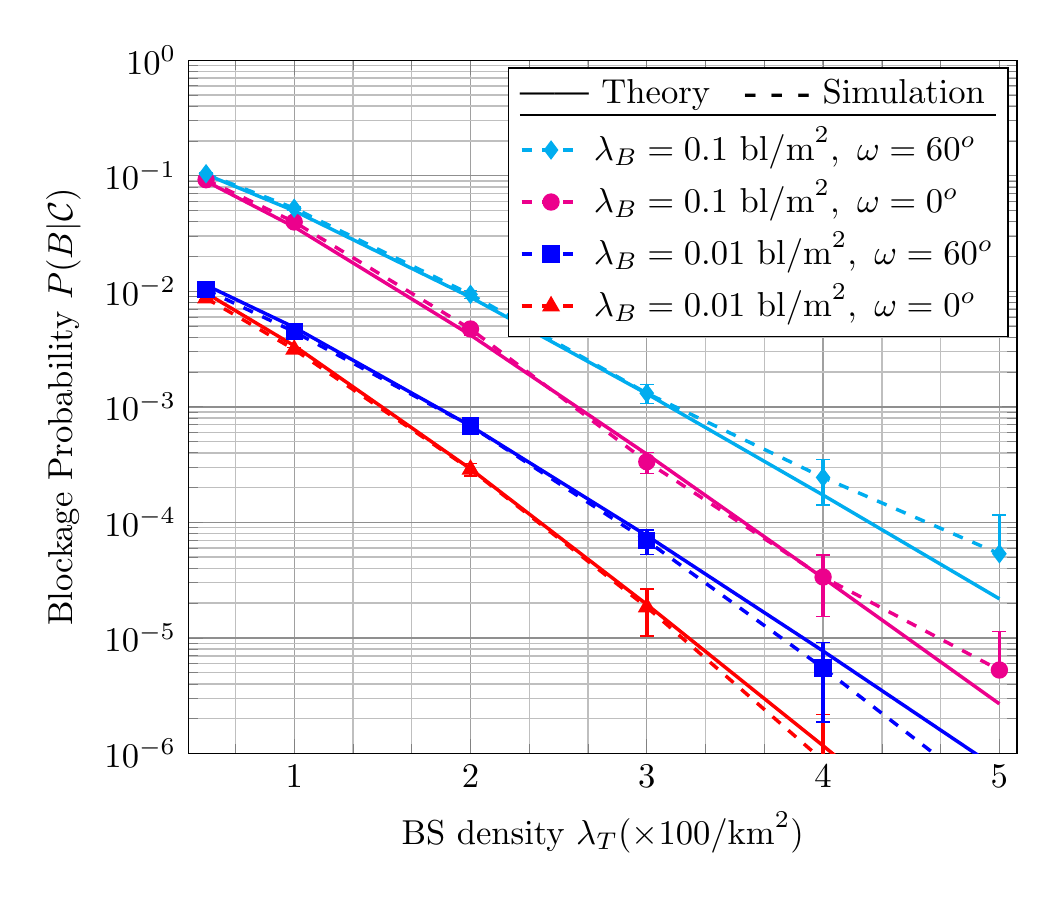}
    \captionsetup{skip=-2pt}
	\caption{Blockage Probability conditioned on coverage event vs BS density for different values of blocker density and self-blockage angle.}
	\label{fig:condBlProb}
	\vspace{-6mm}
\end{figure}
\begin{figure}[!t]
	\centering
	\includegraphics[width=0.435\textwidth]{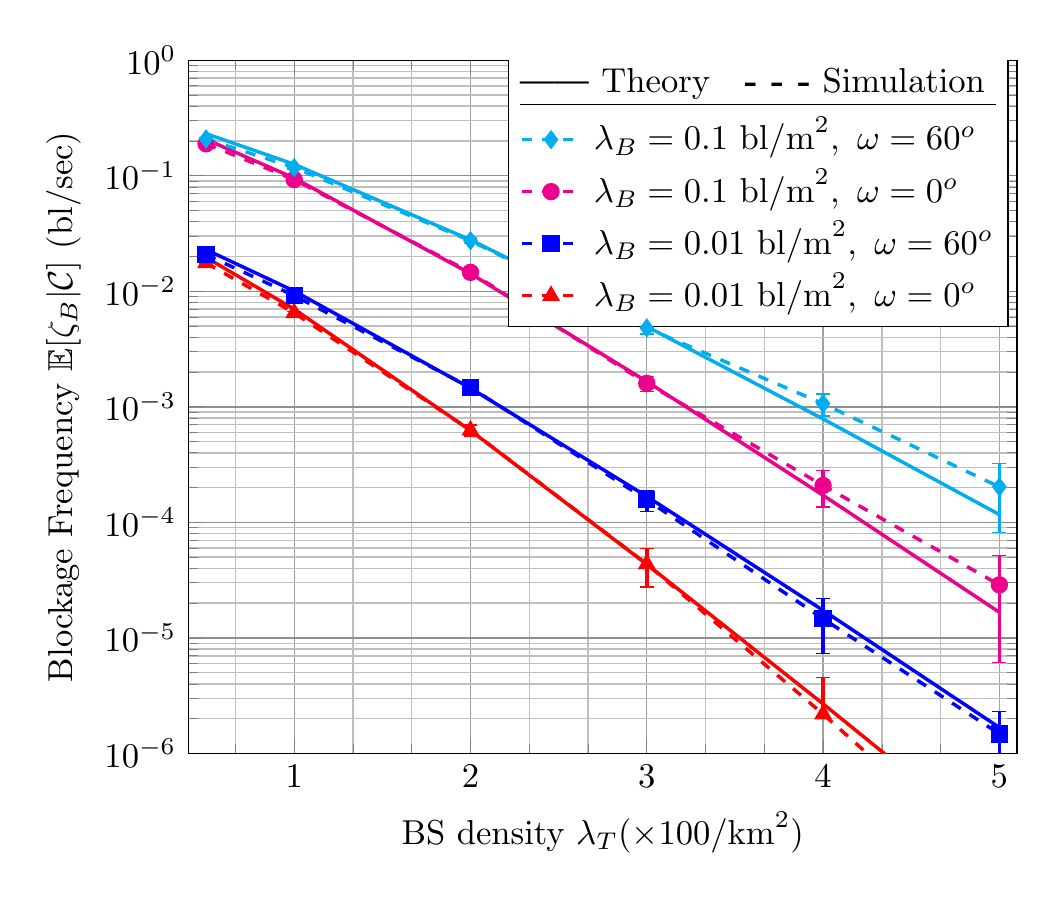}
    \captionsetup{skip=-2pt}
	\caption{Expected blockage frequency conditioned on coverage event vs BS density.}
	\label{fig:condBlfrac}
	\vspace{-5mm}
\end{figure}
\begin{figure}[!t]
	\centering
	\includegraphics[width=0.42\textwidth]{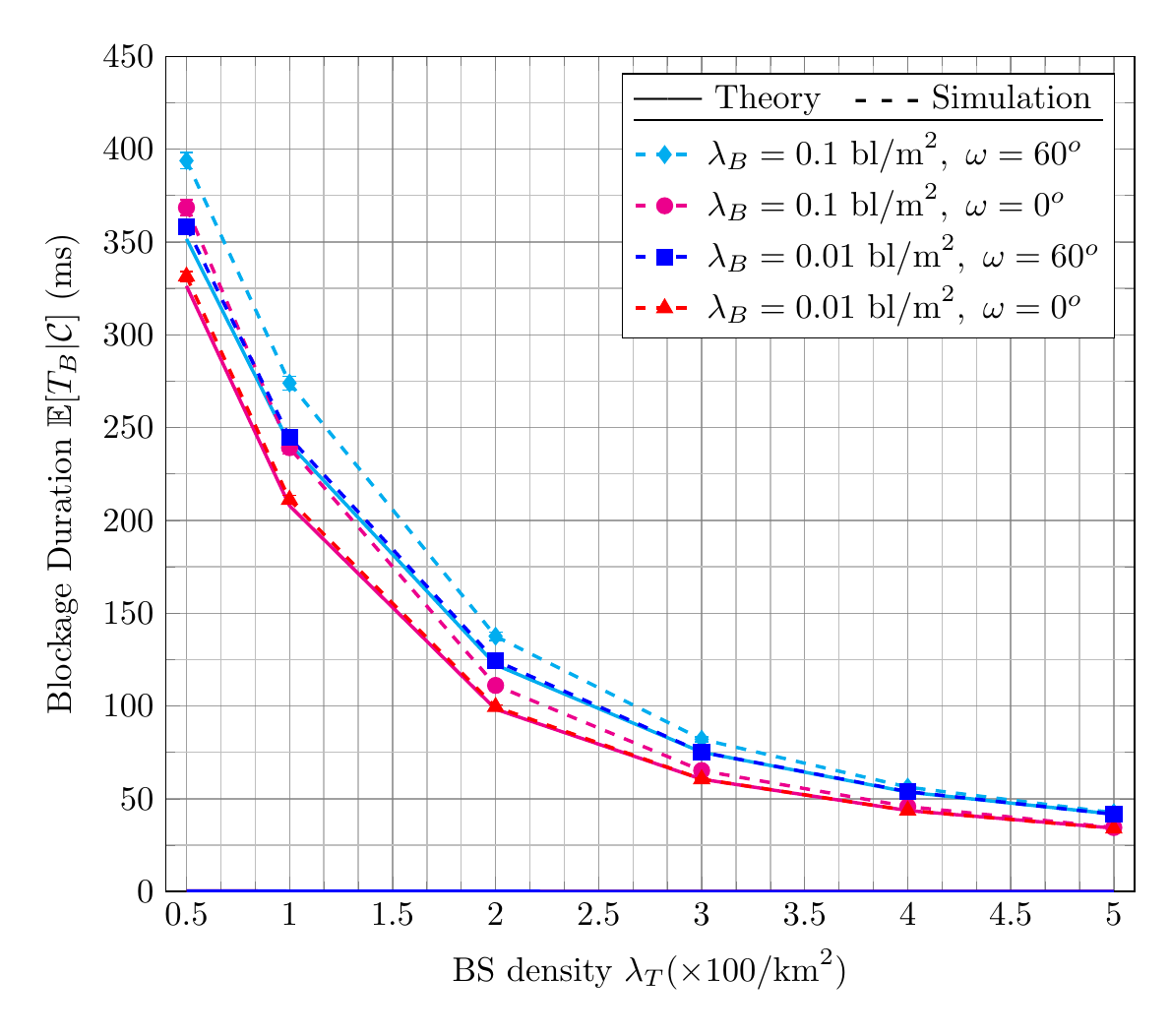}
    \captionsetup{skip=-2pt}
	\caption{Expected blockage duration conditioned on coverage event vs BS density. Note that the theoretical blockage duration is the same for different blocker densities for a fixed self-blockage angle $\omega$.}
	\label{fig:condBldur}
	\vspace{-5mm}
\end{figure}
\begin{figure}[!t]
	\centering
	\includegraphics[width=0.44\textwidth]{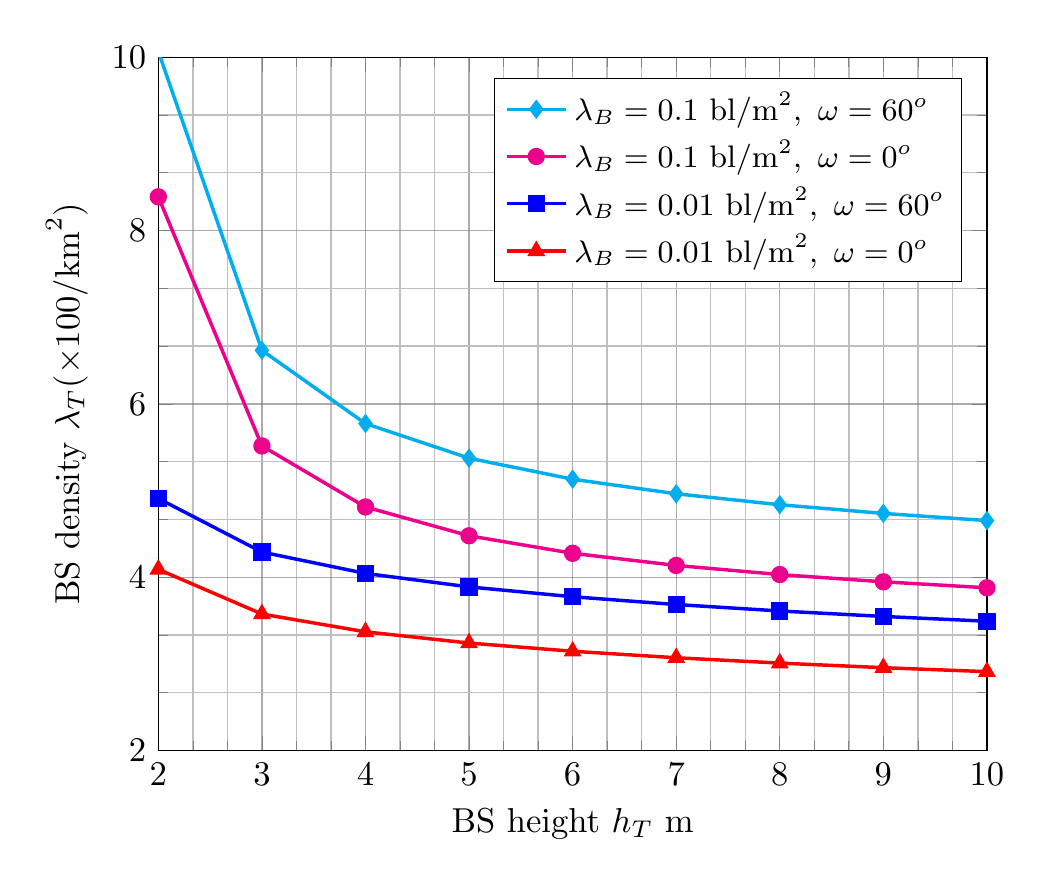}
    \captionsetup{skip=-2pt}
	\caption{The trade-off between BS height and density for fixed blockage probability $P(B|\setC)=1e-5$. }
	\label{fig:heightvsdensity}
	\vspace{-6mm}
\end{figure}

\section{Numerical Evaluation}
\label{sec:numResults}
This section compares our analytical results with MATLAB simulation\footnote{Our simulator MATLAB code is available at github.com/ishjain/mmWave.} 
where the movement of blockers is generated using the random waypoint mobility model~\cite{randomWayPoint,randomWayPointSim}. For the simulation, we consider a rectangular box of $200 \ \text{m}\times 200 \ \text{m}$ and blockers are located uniformly in this area.
Our area of interest is the disc $B(o,R)$ of radius $R=100 $m, which perfectly fits in the considered rectangular area. 
The blockers chose a direction randomly, and move in that direction for a time-duration of $t\sim \text{Unif}[0,60] \ \text{sec}$. To maintain the density of blockers in the rectangular region, we consider that once a blocker reaches the edge of the rectangle, they get reflected.

We consider two values of mobile blocker density, $0.01$ and $0.1$ bl/m$^2$, and two values of the self-blocking angle $\omega$ ($0$ and $\pi/3$) for our study. Figures~\ref{fig:condBlProb},~\ref{fig:condBlfrac}, and~\ref{fig:condBldur} present the statistics of blockages when the UE has at least one serving BS, \textit{i.e.}, the UE is always in the coverage area of at least one BS. From Figure~\ref{fig:condBlProb} and Figure~\ref{fig:condBlfrac}, we can observe that the blockage probability and the expected blockage frequency decrease exponentially with BS density. From the point of view of interactive applications such as AR/VR, video conferencing, online gaming, and others, this means that a higher BS density can potentially decrease interruptions in the data transmission. For example, for a blocker density of 0.1 bl/m$^2$, a BS density of $100$/km$^2$ can decrease the interruptions to once in ten seconds, $200$/km$^2$ can decrease them to once in 100 seconds, and $300$/km$^2$ 
decrease them to once in 1000 seconds. Reducing the frequency of interruptions is particularly crucial for AR/VR applications, therefore from this perspective a density of 200-300/km$^2$ may be required. This corresponds to about 6 to 9 BS, respectively,  within the range of each UE. From Figure~\ref{fig:condBldur}, we can observe that caching of $100$ ms worth of data is required for a BS density $200$/km$^2$ to have uninterrupted services. For AR and tactile applications, caching is not a solution and a delay of 100 ms may be an unacceptable delay. Switching to microwave networks such as 4G during blockage events may be an alternative solution instead of deploying a high  BS density, but then this may need careful network planning so as to not overload the 4G network. The amount of required cached data decreases with increasing BS density. A BS density of $300$/km$^2$ and $500$/km$^2$ can bring down the required cached data to $60$ ms and $40$ ms, respectively. This may be acceptable for AR/VR applications if these freezes are infrequent. Thus, the cellular architecture needs to consider the optimal amount of cached data and the optimal BS density needed to mitigate the effect of these occasional high blockage durations to satisfy QoS requirements for AR/VR application without creating nausea. A tentative conclusion is that perhaps a minimum acceptable density of 300/km$^2$ (which corresponds to about 9 BS within range of each UE) is needed to keep interruptions lasting about 60 ms to once every 1000 seconds.    

We also observe that both simulation and analytical results are approximately the same for a low blocker density of $0.01$ bl/m$^2$. 
From Figure~\ref{fig:condBldur}, we observe our analytical result deviates from the simulation result for lower values of BS densities. However, the percentage error ($\sim10-15\%$) is not significant. 



\begin{table}[!t]
\vspace{2mm}
\caption{Simulation parameters}
\label{SimParams}
\centering
\begin{tabular}{|c|c|}\hline
	\textbf{Parameters} & \textbf{Values} \\\hline 
    Radius  $R$ &100 m \\
    Velocity of blockers $V$&1 m/s\\
    Height of Blockers $h_B$ & 1.8 m\\
    Height of UE $h_R$ & 1.4 m\\
    Height of APs $h_T$ & 5 m\\
    Expected blockage duration $1/\mu$ &1/2 s\\
   Self-blockage angle $\omega$ & 60$^\text{o}$ \\ 
\hline
\end{tabular}
\vspace{-4mm}
\end{table}


5G-PPP has issued requirements for 5G use cases~\cite{mohr20165g} with  service reliability $\geq 99.999\%$ for specific mission-critical services.   
From Figure \ref{fig:condBlProb}, we can infer that the minimum BS density required for a maximum blockage probability $P(B| \setC) = 1e-5$ is 400 BS/km$^2$ for a blocker density of 0.01 bl/m$^2$ and self-blockage angle of $60^o$. For a higher blocker density, the required BS density increases linearly.
Note that this again imposes a higher BS density than would be necessary from most models based solely on capacity needs (roughly 100 BS/km$^2$~\cite{CellularCap-Rap}).

The BS height vs. density trade-off is shown in Figure \ref{fig:heightvsdensity}. Note, for example, that doubling the height of the BS from 4m to 8 m reduces the BS station density requirement by approximately 20\%. 
The optimal BS height and density can be obtained by performing a cost analysis.


\section{Conclusion}
\label{sec:conclusion}
In this paper, we presented a simplified model for the key QoS parameters such as blockage probability, frequency, and duration in mmWave cellular systems. Our model considered an open park-like area with dynamic blockage due to mobile blockers and self-blockage due to user's own body. We have not considered Non-Line-of-Sight (NLOS) paths for such an environment, but it will be included in future work. The user is considered blocked when all potential BSs around the UE are blocked simultaneously.
We verified our theoretical model with MATLAB simulations considering the random waypoint mobility model for blockers. Our results tentatively show that the density of BS required to provide acceptable quality of experience to AR/VR applications is much higher than that obtained by capacity requirements alone. This suggests that the mmWave cellular networks may be blockage limited instead of capacity limited. 
We also plan to study blockage effects in networks with hexagonal cells in our future work.

\bibliographystyle{IEEEtran}
\bibliography{references}

\appendix

\subsection{Proof of Theorem 1}\label{app:th_prob}
We first derive $P(B|N)$ in (\ref{eqn:pBgivenN}) as
\begin{equation}\label{eqn:proof1-firsthalf}
\begin{split}
&P(B|N)=\int\!\!\!\int_{r_i}\!\! P(B|N,\{R_i\})\;f(\{R_i\}|N)\; dr_1\cdots dr_n\\
&=\int\!\!\!\int_{r_i}\prod_{i=1}^n \frac{(C/\mu)r_i}{1+(C/\mu)r_i} \frac{2r_i}{R^2} dr_i \\
& = \prod_{i=1}^n \int_{r=0}^{R} \frac{(C/\mu)r}{1+(C/\mu)r} \frac{2r}{R^2} dr \\
&=\left(\int_{r=0}^R\frac{(C/\mu)r}{1+(C/\mu)r} \frac{2r}{R^2}\,dr\right)^n\\
&= \left(\int_{r=0}^R\left(\frac{2r}{R^2}-\frac{2\mu}{R^2C}+\frac{2\mu}{R^2C}\frac{1}{(1+Cr/\mu)}\right)\,dr\right)^n\\
&=\left(\left(\frac{r^2}{R^2}-\frac{2\mu r}{R^2C}+\frac{2\mu^2}{R^2C^2}\log(1+Cr/\mu)\right)\bigg|_0^R\right)^n\\
&=\left(1-\frac{2\mu }{RC}+\frac{2\mu^2}{R^2C^2}\log(1+RC/\mu)\right)^n\\& = (1-a)^n,
 \end{split}
\end{equation}
where $a$ is given in (\ref{eqn:a}).
Next, we evaluate $P(B)$ in (\ref{eqn:ep1n}) as
\begin{equation*}\label{eqn:probPBpart2Dyn}
\begin{split}
   &P(B)=\sum_{n=0}^{\infty} P(B|N)P_N(n)\\
   &=\sum_{n=0}^\infty (1-a)^n  \frac{[p\lambda_{T} \pi R^2]^n}{n!}e^{-p\lambda_{T} \pi R^2} \\
& = e^{-ap\lambda_{T} \pi R^2}\sum_{n=0}^\infty \frac{[(1-a)\lambda_{T} \pi R^2]^n}{n!}e^{-(1-a)\lambda_{T} \pi R^2}\\
&= e^{-ap\lambda_{T} \pi R^2}. \\
\end{split}
\end{equation*}

Finally, the conditional blockage probability $P(B|\mathcal{C})$ conditioned on coverage event is obtained as
\begin{equation}\label{eqn:pureBlk}
\begin{split}
P(B|\setC)& = \frac{P(B,\setC)}{P(\setC)} =\frac{\sum_{n=1}^{\infty} P(B|N)P_N(n)}{P(\setC)}\\
&=\frac{e^{-a p\lambda_T\pi R^2} - e^{-p\lambda_T\pi R^2}}{1-e^{-p\lambda_T\pi R^2}}.
\end{split}
\end{equation}

This concludes the proof of Theorem~\ref{th1}.



\subsection{Proof of theorem 2}\label{app:th_freq}
We prove theorem 2 here.
We first evaluate $\E[\zeta_B|N]$ given in (\ref{eqn:blockageDurEQn1}) as
\begin{equation}
\begin{split}
\E[\zeta_B|N] &= n\mu \!\!\int\!\!\!\int_{r_i} \prod_{i=1}^n \frac{(C/\mu)r_i}{1+(C/\mu)r_i} \frac{2r_i}{R^2} dr_i \\
& = n\mu (1-a)^n,
\end{split} 
\end{equation}
which follows similar to (\ref{eqn:proof1-firsthalf}). Next, we evaluate $\E[\zeta_B]$ given in (\ref{ep1n}) as
\begin{equation}
\begin{split}
 &\mathbb{E} \left[\zeta_B\right]  =  \sum_{n=0}^\infty  n \mu(1-a)^n \frac{[p\lambda_{T} \pi R^2]^n}{n!}e^{-p\lambda_{T} \pi R^2} \\
 &= \mu(1-a)p\lambda_{T} \pi R^2e^{-ap\lambda_{T} \pi R^2}\times\\
 &\qquad\quad\sum_{n=0}^\infty  \frac{[(1-a)p\lambda_{T} \pi R^2]^{(n-1)}}{(n-1)!}e^{-(1-a)p\lambda_{T} \pi R^2} \\
 &= \mu(1-a)p\lambda_{T} \pi R^2e^{-ap\lambda_{T} \pi R^2}.\\
\end{split}
\end{equation}

Finally, the expected frequency of blockage conditioned on the coverage events (\ref{eqn:coveragep}) is given by

\begin{equation}
\begin{split}
\MoveEqLeft \mathbb{E} \left[\zeta_B|\setC\right] = \frac{\sum_{n=1}^\infty \E[\zeta|N] P_N(n)}{P(\setC)} =\frac{\sum_{n=0}^\infty \E[\zeta|N] P_N(n)}{P(\setC)} \\
 &= \frac{\mu(1- a)p\lambda_T\pi R^2e^{-ap\lambda_T\pi R^2}}{{1-e^{-p\lambda_T\pi R^2}}}.
\end{split}
\end{equation} 
This concludes the proof of Theorem~\ref{th2}. 
\subsection{Proof of theorem \ref{th3}}\label{app:th_dur}
Using the results from (\ref{eqn:exp_TBgivenN}), we find the expected blockage duration $\E[T_B|\setC]$ conditioned on the coverage event $\setC$ defined in (\ref{eqn:coveragep}) as follow
\begin{equation}
\begin{split}
 \mathbb{E} \left[T_B|\setC\right] &= \frac{\sum_{n=1}^\infty \frac{1}{n\mu} P_N(n)}{P(\setC)} \\
&= \frac{\sum_{n=1}^\infty \frac{1}{n\mu} \frac{[p\lambda_{T} \pi R^2]^n}{n!}e^{-p\lambda_{T} \pi R^2}}{1-e^{-p\lambda_T\pi R^2}} \\
&= \frac{e^{-p\lambda_{T} \pi R^2}}{\mu\left(1-e^{-p\lambda_T\pi R^2}\right)}\sum_{n=1}^\infty \frac{[p\lambda_{T} \pi R^2]^n}{n n!}.
\end{split}
\end{equation}
This concludes the proof of Theorem \ref{th3}.


\subsection{Proof of approximation of expected duration}\label{app:approx_dur}
The expectation of a function $f(n) = 1/n$ can be approximated using the Taylor series as
\begin{equation}
\begin{split}
\E[f(n)] &= \E(f(\mu_n+(x-\mu_n))),\\
&=\E[f(\mu_n)+f'(\mu_n)(x-\mu_n)+\frac{1}{2}f''(\mu_n)(x-\mu_n)^2]\\
&\approx f(\mu_n)+\frac{1}{2}f''(\mu_n)\sigma_n^2\\
&=\frac{1}{\mu_n}+\frac{\sigma_n^2}{\mu_n^3},
\end{split}
\end{equation}
where $\mu_n$ and $\sigma_n^2$ are the mean and variance of Poisson random variable $N$ given in (\ref{eqn:PN}). We get the required expression by substituting $\mu_n = p\lambda_T\pi R^2$ and $\sigma_n^2 = p\lambda_T\pi R^2$. 
\end{document}